\documentclass[reqno,12pt,letterpaper]{amsart}
\usepackage{amsmath,amssymb,amsthm,graphicx,mathrsfs,url}
\usepackage[usenames,dvipsnames]{color}
\usepackage[colorlinks=true,linkcolor=Red,citecolor=Green]{hyperref}
\usepackage{amsxtra}

\def\?[#1]{\textbf{[#1]}\marginpar{\Large{\textbf{??}}}}

\let\epsilon=\varepsilon 
\let\phi=\varphi

\newcommand{\HH}{{\mathbb H}}

\newcommand{\RR}{{\mathbb R}}

\newcommand{\ZZ}{{\mathbb Z}}

\newcommand{\NN}{{\mathbb N}}
\newcommand{\CC}{{\mathbb C}}
\newcommand{\CI}{{{\mathcal C}^\infty}}

\newcommand{\dCI}{{\dot{\mathcal C}^\infty}}

\newtheorem{prop}{Proposition}

\numberwithin{equation}{section}
\numberwithin{prop}{section}

\DeclareMathOperator{\Res}{Res}

\DeclareMathOperator{\comp}{comp}

\let\Im=\Imag

\let\Re=\Real

\DeclareMathOperator{\vol}{vol}

\title{Complex Higgs Oscillators}
\author{Haoren Xiong}
\email{xiong@math.berkeley.edu}
\address{Department of Mathematics, University of California,
Berkeley, CA 94720, USA}

\begin{document}

\begin{abstract}
    In this note we discuss the complex version of the Higgs oscillator on the hyperbolic space. The eigenvalues and resonances of the complex Higgs oscillator are computed in different examples in the hyperbolic setting. We also propose open problems like whether the complex absorbing potential (CAP) method works for asymptotically hyperbolic manifolds and finding hyperbolic analogues of the complex harmonic oscillator. 
\end{abstract}

\maketitle

\section{Introduction}

The Higgs oscillator \cite{HiggsPaper} (see also Pallares-Rivera--Kirchbach \cite{Higgs}) is considered as an analogue of the quantum harmonic oscillator on the hyperbolic plane. In this paper we discussed its complex version, in analogy to the complex harmonic oscillator in the Euclidean space studied by, among others, Davies \cite{Dav}.

Our motivation comes from the complex absorbing potential (CAP) method which has been used as a computational
tool for finding scattering resonances -- see Riss--Meyer \cite{RiMe} and Seideman--Miller \cite{semi} for an early treatment and Jagau et al \cite{Jag} for some recent developments. Zworski \cite{Zw-vis} showed that scattering resonances of $-\Delta+V$, $V \in L^\infty_{\comp}(\RR^n) $, are limits of eigenvalues of $-\Delta+V-i\epsilon x^2$ as $\epsilon\to 0+$ . That was extended to potentials which are exponentially decaying in \cite{xiong2020}. In addition, in \cite{xiong2021} the author extended it to black box non-compactly supported dilation analytic perturbations.

It is natural to ask if the same method works for the hyperbolic space and more generally, asymptotically hyperbolic manifolds. We can formulate the problem as follows: let $(M,g)$ be a complete Riemannian manifold of dimension $n+1$ with boundary $\partial M$ given by $\{\rho=0\}$ where $\rho: \overline{M}\to [0,\infty)$ is a $\CI$ function such that $d\rho\neq 0$ on $\partial M$, and $\rho>0$ on $M$. We assume that the metric $\rho^2 g$ extends to a smooth Riemannian metric on $\overline{M}$ and that $|d\rho|_{\rho^2 g} = 1$ on $\partial M$. Let $\Delta_g\geq 0$ be the Laplace--Beltrami operator for the metric $g$. Since the spectrum is contained in $[0,\infty)$ the operator $\Delta_g - n^2/4 - \lambda^2$ is invertible on $L^2(M,d\vol_g)$ for $\Im\lambda > n/2$. Hence we can define
\[
    R(\lambda) := (\Delta_g - n^2/4 - \lambda^2)^{-1} : L^2(M,d\vol_g) \to H^2(M,d\vol_g),\quad \Im\lambda>n/2. 
\]
Let $\dCI(M)$ denote functions which are extendable to smooth functions supported in $\overline{M}$. We note that elliptic regularity shows that $ R ( \lambda ) : \dCI ( M ) \to
\CI ( M )$, $ \Im\lambda  > n/2 $. The resolvent $R(\lambda):\dCI(M)\to \CI(M)$ continues meromorphically from $ \Im\lambda > n/2 $ to $ \CC$ with poles of finite rank, see Mazzeo--Melrose \cite{mm}, Guillarmou \cite{g}, Guillop\'e--Zworski \cite{GuZw} and Vasy \cite{vasy1},\cite{vasy2}. We denote the poles by $\Res(\Delta_g)=\{\lambda_j\}_{j=1}^\infty$. Does there exist a function $f$ such that the operator $\Delta_g - n^2/4 - i\epsilon f$, $\epsilon>0$, has discrete $L^2(M,d\vol_g)$ spectrum $\{\lambda_j(\epsilon)^2\}$ and that 
\[
    \lambda_j(\epsilon) \to \lambda_j,\quad\textrm{as } \epsilon\to 0+ .
\]
We remark that for finite volume surface with cusp ends this holds with $f(x)=d(x,x_0)^2$ where $d$ is the hyperbolic distance -- see \cite[Example 3]{xiong2021}.

In this note we want to propose this as an open problem. We first compare the situation to the
Euclidean case and then discuss hyperbolic analogues of the complex harmonic oscillator. In fact, the first obstacle is the lack of an analogue of the complex harmonic oscillator in the hyperbolic setting. In \cite{Zw-vis} and \cite{xiong2020}, the operator $\Delta-i\epsilon x^2$, $\epsilon>0$ plays an important role as it is an unbounded operator on $L^2(\RR^n)$ with a discrete spectrum given by $\{ e^{-i\pi/4} \sqrt{\epsilon}(2|\alpha|+n) : \alpha\in\mathbb{N}_0^n\}$, $|\alpha|:=\alpha_1+\cdots+\alpha_n$, see \cite{Dav}. On a hyperbolic manifold $(M,g)$, we aim to find a complex-valued function $f\in\CI(M)$ such that $\Delta_g + f$ is an operator on $L^2(M,d\vol_g)$ with discrete spectrum. Ideally, we should also require $f$ to be unbounded near infinity like the function $-i\epsilon x^2$ in the Euclidean case, which would provide the compactness of the resolvent $(\Delta_g - n^2/4 + f -z)^{-1} : L^2(M,d\vol_g)\to L^2(M,d\vol_g)$. However, it is hard to find a function $f$ satisfying all the requirements above. We will explore a candidate $f=\omega^2 \tanh^2 r$ where $\omega\in\CC$, $r$ is the hyperbolic radius. The operator $\Delta_g + \omega^2\tanh^2 r$ is called the Higgs oscillator in the hyperbolic space, whose spectrum and resonances can be explicitly computed, see \cite{Higgs} for more details. The drawback of this candidate is the boundedness of $f$ thus we lose the compacteness of the resolvent $(\Delta_g - n^2/4 + f -z)^{-1}$. We remark that it is still an open problem to find an ideal analogue of the complex harmonic oscillator in the hyperbolic setting. We hope that the following introduction could popularize this natural problem.

\section{Preliminaries: P\"oschl--Teller potentials}
\label{PT}

We recall the following definition from P{\"o}schl--Teller \cite{Teller}: the P{\"o}schl--Teller potential is defined on $\RR$ by
\[
    V_{\mu,\nu}(r) := \frac{\mu(\mu+1)}{\sinh^2 r} - \frac{\nu(\nu+1)}{\cosh^2 r},\quad r\in\RR,
\]
$V_{\mu,\nu}$ is a real potential if $\mu, \nu$ are taken in $-1/2 + i(0,\infty) \cup [-1/2,\infty)$. In this section we will focus on the case in which $V_{\mu,\nu}$ is complex-valued and review some properties of the Hamiltonian $D_r^2 + V_{\mu,\nu}(r)$ on the half line $(0,\infty)_r$. The following result is based on the analysis of $D_r^2 + V_{\mu,\nu}(r)$ in Guillop\'e--Zworski \cite[Appendix]{GuZw2}:
\begin{prop}
\label{prop:PTspectrum}
    The Schr\"odinger operator $D_r^2 + V_{\mu,\nu}$ (respectively $D_r^2 + V_{0,\nu}$) has $\RR^+$ as continuous spectrum. The determinant of the scattering matrix for $D_r^2 + V_{\mu,\nu}$ is given by the reflection coefficient
    \begin{equation}
    \label{eqn:scattering matrix}
        s_{\mu,\nu}^{PT} (k) = -\frac{\Gamma(ik)\Gamma((\mu+\nu-ik)/2 + 1)\Gamma((\mu-\nu-ik+1)/2)2^{-ik}}{\Gamma(-ik)\Gamma((\mu+\nu+ik)/2 + 1)\Gamma((\mu-\nu+ik+1)/2)2^{ik}},
    \end{equation}
    and for $D_r^2 + V_{0,\nu}$ given by 
    \begin{equation}
    \label{eqn:scattering matrix2}
        s_\nu^{PT} (k) = -\frac{\Gamma(ik)^2 \Gamma(\nu-ik + 1)\Gamma(-\nu-ik)}{\Gamma(-ik)^2 \Gamma(\nu + ik + 1)\Gamma(-\nu + ik)}.
    \end{equation}
    
    The Schr\"odinger operator $D_r^2 + V_{\mu,\nu}$ (resp. $D_r^2 + V_{0,\nu}$) has non-empty discrete spectrum if and only if $\Re (\nu-\mu) > 1$ (resp. $\Re\nu > 0$). The discrete spectrum is given by
    \[
        \begin{split}
            \sigma_d (D_r^2 + V_{\mu,\nu}) &= \{-(\nu-\mu-1-2n)^2 : n\in\NN,\, 2n<\Re(\nu-\mu-1)\}, \\
            \sigma_d (D_r^2 + V_{0,\nu}) &=
            \{-(\nu-n)^2 : n\in\NN,\,n<\Re\nu\}.
        \end{split}
    \]
\end{prop}

\begin{proof}
Through a conjugation by $\sinh^{\mu+1} r \cosh^{\nu +1} r$ and the change of variable $u=-\sinh^2 r$, the Schr{\"o}dinger equation
\begin{equation}
\label{eqn:Vmunu}
    D_r^2 \psi + V_{\mu,\nu} \psi - k^2 \psi = 0
\end{equation}
is reduced to the hypergeometric equation
\[
\begin{split}
    u(1-u)F''(u) &+ [(\mu+3/2) - (\mu+\nu+3)u]F'(u) \\
    &- [((\mu+\nu+2)/2)^2 + (k/2)^2]F = 0.
\end{split}
\]
The Schr{\"o}dinger equation \eqref{eqn:Vmunu} has the following independent solutions (if $\mu\neq -\frac{1}{2}$):
\begin{equation}
\label{eqn:Emunu}
    \begin{split}
        E_{\mu,\nu}(k)(r) = &\sinh^{1+\mu}r \cosh^{1+\nu}r \\
        &\times {_{2}F_1} ( (\mu+\nu-ik+2)/2,(\mu+\nu+ik+2)/2,\mu+\frac{3}{2};-\sinh^2 r ),
    \end{split}
\end{equation}
\begin{equation}
\label{eqn:Fmunu}
    \begin{split}
        F_{\mu,\nu}(k)(r) = &\sinh^{-\mu}r \cosh^{1+\nu}r \\
        &\times {_{2}F_1} ( (-\mu+\nu-ik+1)/2,(-\mu+\nu+ik+1)/2,\frac{1}{2}-\mu;-\sinh^2 r ).
    \end{split}
\end{equation}
The asymptotic expansion of \eqref{eqn:Emunu} at infinity is given, if $ik$ is not an integer, by
\begin{equation}
\label{eqn:EmunuAsymp}
    \begin{split}
        E_{\mu,\nu}(k)(r) \approx\, &\frac{\Gamma(\mu+ 3/2)\Gamma(ik)}{\Gamma((\mu+\nu+ik+2)/2)\Gamma((\mu-\nu+ik+1)/2)}\coth^{\nu+1}r \sinh^{ik}r \\
        &\times {_{2}F_1}( (\mu+\nu-ik+1)/2,(-\mu+\nu-ik+1)/2,1-ik;-\sinh^{-2} r ) \\
        &\frac{\Gamma(\mu+ 3/2)\Gamma(-ik)}{\Gamma((\mu+\nu-ik+2)/2)\Gamma((\mu-\nu-ik+1)/2)}\coth^{\nu+1}r \sinh^{-ik}r \\
        &\times {_{2}F_1}( (\mu+\nu+ik+1)/2,(-\mu+\nu+ik+1)/2,1+ik;-\sinh^{-2} r ),
    \end{split}
\end{equation}
recalling the definition of reflection coefficient for potential scattering (see for instance Dyatlov--Zworski \cite[\S 2.4]{res}), we obtain \eqref{eqn:scattering matrix}.

The potential $V_{0,\nu}$ is smooth on $\RR$, the operator $D_r^2 + V_{0,\nu}$ can be decomposed as the sum of the Dirichlet $(H_\nu^D)$ and Neumann $(H_\nu^N)$ extensions of $D_r^2 + V_{0,\nu}$. The eigenfunctions of the spectral resolution of $H_\nu^N$ are the $F_{0,\nu}(k)$ from \eqref{eqn:Fmunu} and a similar asymptotic expansion at infinity to \eqref{eqn:EmunuAsymp} gives the reflection coefficient $s(H_\nu^N)$. The scattering coefficient $s_\nu^{PT} (k)$ \eqref{eqn:scattering matrix2} is then the product $s_{0,\nu}^{PT} (k) s(H_\nu^N)(k)$.

The asymptotic properties of the eigenfunctions $\eqref{eqn:Emunu}$ and \eqref{eqn:Fmunu} determine the discrete spectra.
\end{proof}

\section{Higgs Oscillator on the Hyperbolic Plane}

We consider the hyperbolic plane $\HH := \{(x,y)\in\RR^2 : y>0\}$ with the metric $y^{-2}(dx^2 + dy^2)$. Instead of coordinates $(x,y)$, we will use the geodesic normal coordinates for hyperbolic metrics. These are coordinates $(r,\varphi)$ for which the $r$-coordinate curves are unit speed geodesics and the $\varphi$-coordinate curves are geodesic circles. The Laplacian $\Delta_{\HH^2} = y^2 (D_x^2 + D_y^2) = D_r^2 - i\coth r D_r + \sinh^{-2}r D_\varphi^2$, where $D_x = i^{-1}\partial_x$, is through conjugation by $\sinh^{1/2} r$, equivalent to the operator
\[
    D_r^2 + \sinh^{-2}r (D_\varphi^2 - 1/4) + 1/4.
\]

We define the complex version of Higgs Oscillator by $\Delta_{\HH^2}+\omega^2 \tanh^2 r$, where $\omega$ is a complex number, which is through the same conjugation as above, equivalent to the operator
\[
    D_r^2 + \frac{D_\varphi^2 - 1/4}{\sinh^2 r} - \frac{\omega^2}{\cosh^2 r} + \omega^2 + \frac{1}{4}
\]
on $L^2((0,\infty)_r \times S_\varphi^1,drd\varphi)$. We can expand this in terms of the eigenfunctions on $S_\varphi^1$ to obtain
\[
    \bigoplus_{m\in\ZZ} D_r^2 + \frac{m^2 - 1/4}{\sinh^2 r} - \frac{\omega^2}{\cosh^2 r} + \omega^2 + \frac{1}{4}.
\]
This leads to the one-dimensional Schr\"odinger operator with P\"oschl--Teller potential $D_r^2 + V_{\mu,\nu}$, where $\mu = |m| - 1/2$ and $\nu=\sqrt{\omega^2 + 1/4} - 1/2$. It follows from Proposition \ref{prop:PTspectrum} that the eigenvalues of $D_r^2 + V_{\mu,\nu}$ are $\{(\nu-\mu-1-2n)^2:n\in\NN,\, 2n<\Re(\nu-\mu-1)\}$. Hence we obtain the discrete spectrum of $\Delta_{\HH^2}+\omega^2\tanh^2 r$:
\[
    \left\{ \omega^2 + \frac{1}{4} - \big( \sqrt{\omega^2 +\frac{1}{4}}-m-1-2n \big)^2 : m,n\in\NN,\, 2n< \Re\sqrt{\omega^2+\frac{1}{4}}-m-1 \right\}.
\]
The scattering matrix \eqref{eqn:scattering matrix} gives the resonances of $\Delta_{\HH^2}+\omega^2\tanh^2 r$: 
\[
    \left\{ \omega^2 + \frac{1}{4} - \big( \sqrt{\omega^2 +\frac{1}{4}}-m-1-2n \big)^2 : m,n\in\NN \right\}.
\]
\begin{figure}
\includegraphics[width=6in]{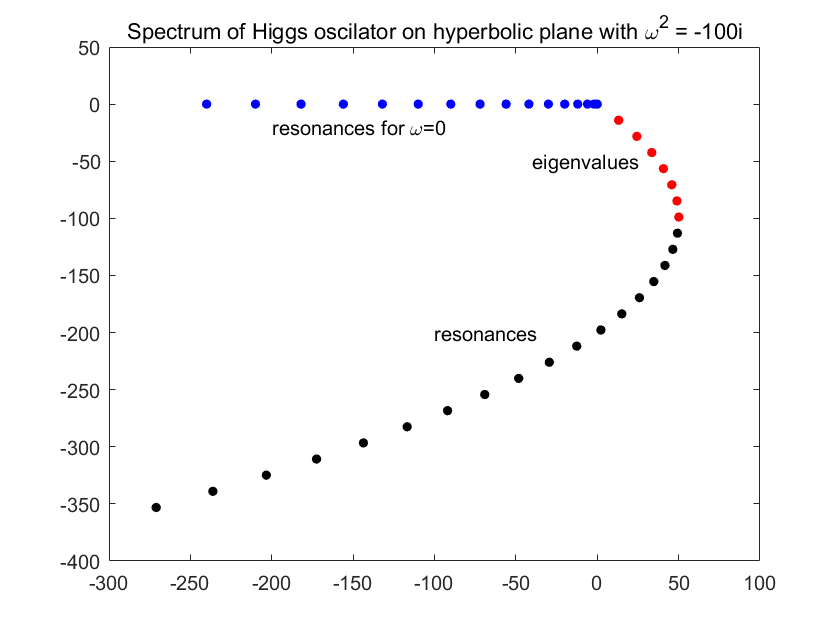}
\caption{The spectrum of the Higgs osicllator on the hyperbolic plane. The red dots are the eigenvalues of $\Delta_{\HH^2}+\omega^2 \tanh^2 r$ with $\omega^2 = -100i$ while the black dots are its resonances. We also plot the resonances of $\Delta_{\HH^2}$, which are the blue dots on the real axis. This shows the deformation of resonances.  }
\end{figure}

\section{Higgs Oscillator with an Eckart barrier}
We consider the one-dimensional Eckart barrier $V=\alpha \cosh^{-2} r$. The Higgs oscillator with an Eckart barrier is given by
\[
    D_r^2 + V +\omega^2 \tanh^2 r = D_r^2 + (\alpha -\omega^2)\cosh^{-2}r + \omega^2.
\]
This can be viewed as a Schr\"odinger operator with P\"oschl--Teller potential $D_r^2 + V_{0,\nu}$ shifted by a constant $\omega^2$, where $\nu = \sqrt{\omega^2 - \alpha + \frac{1}{4}}-\frac{1}{2}$. It follows from Proposition \ref{prop:PTspectrum} that the discrete spectrum of the Higgs oscillator with a Eckart barrier is given by
\[
    \left\{ \omega^2 - \big( \sqrt{\omega^2 -\alpha +\frac{1}{4}} - \frac{1}{2} - n \big)^2 : n\in\NN,\, n< \Re\sqrt{\omega^2+\frac{1}{4}} - \frac{1}{2} \right\}.
\]
The scattering matrix \eqref{eqn:scattering matrix2} gives the resonances:
\[
    \left\{ \omega^2 - \big( \sqrt{\omega^2 -\alpha +\frac{1}{4}} - \frac{1}{2} - n \big)^2 : n\in\NN\right\}.
\]

\section{Higgs Oscillator on hyperbolic half-cylinder}
We consider the hyperbolic half-cylinder $Y_{0l}\simeq (0,\infty)_r\times (\RR/l\ZZ)_\theta$ with metric $dr^2+\cosh^2 r d\theta^2$. The Laplacian (Dirichlet boundary condition) $\Delta_{Y_{0l}}=D_r^2 - i\tanh r D_r + \cosh^{-2}\Delta_{\RR/l\ZZ}$ is, through a conjugation by $\cosh^{1/2}r$, equivalent to the operator
\[
    D_r^2 + \frac{\Delta_{\RR/l\ZZ}+1/4}{\cosh^2 r} +\frac{1}{4}.
\]
The Higgs oscillator $\Delta_{Y_{0l}}+\omega^2\tanh^2 r$, is then equivalent to the operator
\[
    D_r^2 - \frac{\omega^2 -\Delta_{\RR/l\ZZ}-1/4}{\cosh^2 r} + \omega^2 + \frac{1}{4},
\]
which admits the following expansion:
\[
     \bigoplus_{m\in\ZZ} D_r^2 - \frac{\omega^2-(2\pi m/l)^2 - 1/4}{\cosh^2 r} + \omega^2 + \frac{1}{4}.
\]
The corresponding one-dimensional one-dimensional Schr\"odinger operator is $D_r^2 + V_{0,\nu}$ on $(0,\infty)$ with Dirichlet boundary condition, where we put $\nu = \sqrt{\omega^2 - (2\pi m/l)^2}-1/2$. Hence by Proposition \ref{prop:PTspectrum} the discrete specrtum of $\Delta_{Y_{0l}}+\omega^2 \tanh^2 r$ is
\[
\begin{gathered}
    \{ \omega^2 + 1/4 - \big( \sqrt{\omega^2 -(2\pi m/l)^2}-2n-3/2 \big)^2 : m\in\ZZ, \\
    n\in\NN,\,2n< \Re\sqrt{\omega^2 -(2\pi m/l)^2}-3/2 \},
\end{gathered}
\]
while the analysis of \eqref{eqn:scattering matrix} gives the resonances:
\[
    \{ \omega^2 + 1/4 - \big( \sqrt{\omega^2 -(2\pi m/l)^2}-2n-3/2 \big)^2 : m\in\ZZ, 
    n\in\NN \}.
\]
\begin{figure}
\includegraphics[width=6in]{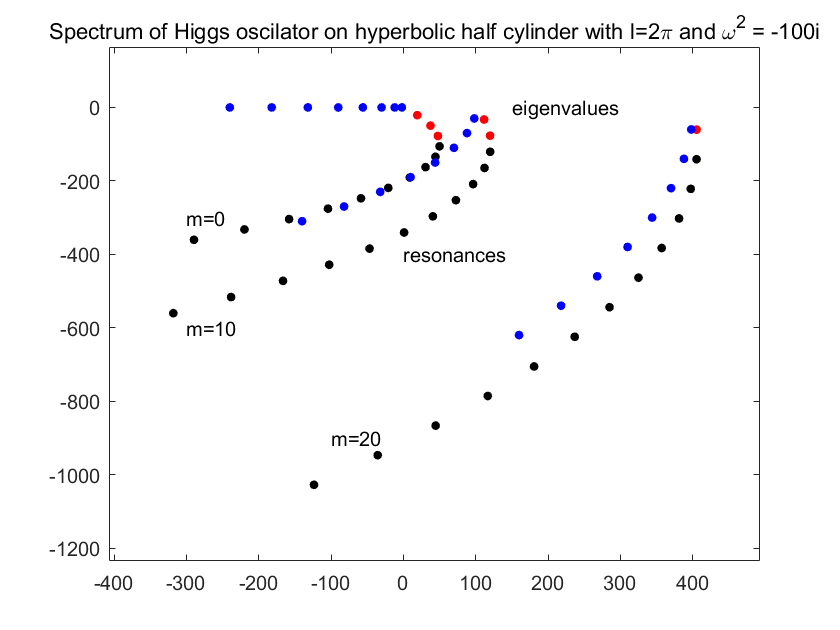}
\caption{The spectrum of the Higgs osicllator on hyperbolic half-cylinder with parameter $l=2\pi$ and $\omega^2 = -100i$. We only plot the spectrum with respect to the Fourier modes $m=0,10,20$, here the red dots are eigenvalues and the black dots are resonances. We also plot resonances for $\omega=0$ with respect to the same Fourier modes to show the deformation of resonances.}
\end{figure}

\section{Discussion}

The explicit formulae show that the resonances in all cases are deformed and some do become eigenvalues. However, in this setting we cannot obtain resonances as limits of these eigenvalues, as one would want for CAP method. In fact, for $\omega$ with small modulus there are no eigenvalues for the complex Higgs oscillators at all. 

\def\arXiv#1{\href{http://arxiv.org/abs/#1}{arXiv:#1}}

\end{document}